\documentclass[runningheads]{llncs}

\usepackage{amsmath}
\usepackage{latexsym}
\usepackage{amsfonts,amssymb}
\usepackage{graphicx, gastex}
\usepackage{cite}
\usepackage{mathrsfs}

\DeclareSymbolFont{rsfscript}{OMS}{rsfs}{m}{n}
\DeclareSymbolFontAlphabet{\mathrsfs}{rsfscript}

\DeclareMathOperator{\dt}{.}
\DeclareMathOperator{\Syn}{Syn}
\begin{document}
\title{Principal ideal languages \\ and synchronizing automata}

\author{Vladimir V. Gusev, Marina I. Maslennikova, Elena V. Pribavkina}

\institute{Ural Federal University, Ekaterinburg, Russia\\
\email{vl.gusev@gmail.com, maslennikova.marina@gmail.com, elena.pribavkina@usu.ru}}

\maketitle

\begin{abstract}
We study ideal languages generated by a single word. We provide an algorithm to construct a strongly connected synchronizing automaton for which such a language serves as the language of synchronizing words. Also we present a compact formula to calculate the syntactic complexity of this language.
\end{abstract}

\section{Introduction}

Let $\mathscr{A}=\langle Q,\Sigma,\delta\rangle$ be a \textit{deterministic finite automaton} (DFA for short), where $Q$ is the \textit{state set}, $\Sigma$ stands for the \textit{input alphabet}, and $\delta: Q\times\Sigma\rightarrow Q$ is the \textit{transition function} defining an action of the letters in $\Sigma$ on $Q$. The action extends in a natural way to an action $Q\times \Sigma^{*}\rightarrow Q$ of the free monoid $\Sigma^{*}$ over $\Sigma$; the latter action is also denoted by $\delta$. When $\delta$ is clear from the context, we will write $q\dt w$ instead of $\delta(q,w)$ for $q\in Q$ and $w\in \Sigma^{*}$. In the theory of formal languages the definition of a DFA usually includes the set $F\subseteq Q$ of \textit{terminal states} and an \textit{initial state} $q_0\in Q$. We will use this definition when dealing with automata as devices for recognizing languages. The language $L\subseteq \Sigma^{*}$ is \textit{recognized} (or \textit{accepted}) by an automaton $\mathscr{A}=\langle Q,\Sigma,\delta, F, q_0\rangle$ if $L=\{w\in\Sigma^{*}\mid \delta(q_0,w)\in F\}$. We also use standard concepts of the theory of formal languages such as regular language, minimal automaton, etc. \cite{Perrin}

A DFA $\mathscr{A}=\langle Q,\Sigma,\delta\rangle$ is called
\emph{synchronizing} if there exists a word $w  \in \Sigma^{*}$ which leaves
the automaton in unique state no matter which state in $Q$ it is started to read: $\delta(q ,w)=\delta(q', w)$ for all $q, q' \in Q$. Such word is said
to be \textit{synchronizing} (or \emph{reset}) for the DFA $\mathscr{A}$. This notion has been widely studied
since the work of Jan \v{C}ern\'{y} \cite{Ce64} in 1964. He conjectured that any synchronizing DFA with $n$ states
possesses a synchronizing word of length at most $(n-1)^2$. This conjecture is widely open and is considered one of 
the most longstanding open problems in the combinatorial theory of finite automata. Various techniques were developed 
to approach this conjecture. For more information on synchronizing automata we refer
the reader to the thorough survey by Mikhail Volkov \cite{Vo_Survey}. In this paper we focus on language theoretic
aspects of the \v{C}ern\'{y} conjecture and related questions.

Recall that a language $L$ over $\Sigma$ is called \textit{ideal} if $L=\Sigma^{*}L\Sigma^{*}$. By $\Syn(\mathscr{A})$ we denote the language of all words synchronizing
$\mathscr{A}$. It is easy to see that $\Syn(\mathrsfs{A})$ is an ideal language.
In what follows we consider only ideal regular languages. It was observed in \cite{SOFSEM} that the minimal deterministic automaton $\mathrsfs{A}_L$ recognizing an ideal regular language $L$ is synchronizing, and $\Syn(\mathrsfs{A}_L)=L$. Interesting question arises: how many states the smallest automaton $\mathscr{A}$ such that $\Syn(\mathrsfs{A})=L$ may have? This question was posed in \cite{SOFSEM} and the notion of reset complexity
was introduced. 
The \textit{reset complexity} $rc(L)$ of an ideal language $L$ is the minimal possible number of states in a synchronizing automaton $\mathscr{A}$ such that $\Syn(\mathscr{A})=L$.
For brevity we will call the corresponding automaton MSA (\textit{minimal synchronizing automaton}). 
The \v{C}ern\'{y} conjecture can be stated in terms of reset complexity as follows. Let $\ell$ be the minimal length of words in 
an ideal language $L$, then $rc(L)\ge \sqrt{\ell}+1$. Even a lower bound $rc(L)\ge \frac{\sqrt{\ell}}{C}$ for some constant $C$
would be a major breakthrough.

From descriptive complexity point of view it is interesting to compare reset complexity with the classical state complexity. 
The \textit{state complexity} $sc(L)$ of a regular language $L$ is the number of states in $\mathrsfs{A}_L$. In \cite{SOFSEM} it was observed that $rc(L)\le sc(L)$.  Also in \cite{SOFSEM} it was shown that in some cases $rc(L)$ can be exponentially smaller than $sc(L)$. In particular, it means that the description of an ideal language $L$ by means of an automaton $\mathscr{A}$ for which $\Syn(\mathscr{A})=L$ can be exponentially more succinct than the ``standard'' description via minimal automaton recognizing $L$. The minimal automaton of an ideal regular language always has a sink state (a state fixed by all letters), whereas the corresponding MSA may be strongly connected, which means that for any two states $p$ and $q$ ($p\neq q$) there exists a word mapping $p$ to $q$.
Automata with the sink state and strongly connected automata are essential for the \v{C}ern\'{y} conjecture, since it was shown in \cite{Vo_CIAA07} that it is enough to prove this conjecture for each of the two classes of automata. Thus, we may ask whether it is always possible to construct a strongly connected synchronizing DFA for which $L$ serves as the language of synchronizing words.

We begin to approach this question by considering principal ideal languages, i.e. ideal languages generated by a single word. A principal ideal language is a partial case of a finitely generated ideal language. The latter languages viewed as languages of synchronizing words were considered in \cite{PribR1} and  \cite{PribR2}.

In section 2 we answer the uniqueness question that was posed in \cite{SOFSEM}. The question is whether the uniqueness of an MSA takes place within the class of strongly connected automata. The answer is negative. For the language $L=\Sigma^{*}a^{n-1}b\Sigma^{*}$ there exist two different strongly connected automata with $n+1$ states over $\Sigma=\{a,b\}$ yielding the minimum of reset complexity for $L$.

In section 3 we provide an algorithm to construct a strongly connected synchronizing automaton whose language of synchronizing words is generated by a single word.

In section 4 we consider some algebraic properties of principal ideal languages. In particular, we establish the connection between the syntactic semigroup of such a language and the transition semigroup of a synchronizing automaton for which this language serves as the language of reset words. Also, we find a compact formula for calculating the syntactic complexity of a principal ideal language.

\section{On uniqueness question of an MSA}

It is well-known that the minimal automaton $\mathrsfs{A}_L$ recognizing a given language $L$ is unique up to isomorphism. The same fact does not hold for an MSA, see \cite{SOFSEM}. But the question, whether the uniqueness takes place within the class of strongly connected automata, remained open. Here we answer this question in the negative.

Consider the language $L=\Sigma^{*}a^{n-1}b\Sigma^{*}$ over $\Sigma=\{a,b\}$. In \cite{SOFSEM} it was shown that $rc(L)=n+1$. We present two different strongly connected automata with $n+1$ states for which $L$ serves as the language of synchronizing words. For clarity the corresponding automata ${\mathrsfs A}_{6}$ and ${\mathrsfs B}_{6}$ with six states are shown on Fig.~\ref{fig Z6}.

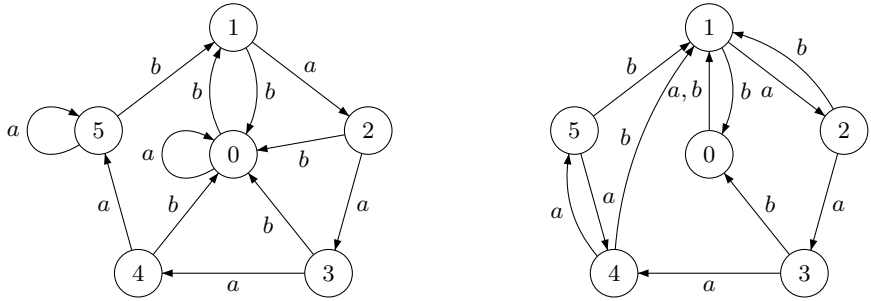
\begin{figure}[ht]
\begin{center}
\unitlength=3pt
\begin{picture}(94,30)
   \gasset{Nw=6,Nh=6,Nmr=3}
\thinlines
\node(A)(0,18){$5$}
\node(B)(17,31){$1$}
\node(C)(34,18){$2$}
\node(D)(29,0){$3$}
\node(E)(5,0){$4$}
\node(F)(17,15){$0$}
\drawedge(A,B){$b$}
\drawedge(B,C){$a$}
\drawedge(C,D){$a$}
\drawedge(D,E){$a$}
\drawedge(E,A){$a$}
\drawedge(C,F){$b$}
\drawedge(D,F){$b$}
\drawedge(E,F){$b$}
\drawloop[loopdiam=6,loopangle=180](A){$a$}
\drawloop[loopdiam=6,loopangle=180](F){$a$}
\drawedge[curvedepth=3](F,B){$b$}
\drawedge[curvedepth=3](B,F){$b$}
\node(A1)(60,18){$5$}
\node(B1)(77,31){$1$}
\node(C1)(94,18){$2$}
\node(D1)(89,0){$3$}
\node(E1)(65,0){$4$}
\node(F1)(77,15){$0$}
\drawedge(A1,B1){$b$}
\drawedge[ELside=r](B1,C1){$a$}
\drawedge(C1,D1){$a$}
\drawedge(D1,E1){$a$}
\drawedge[curvedepth=3](E1,A1){$a$}
\drawedge[curvedepth=-3,ELside=r](C1,B1){$b$}
\drawedge(F1,B1){$a,b$}
\drawedge[curvedepth=3](B1,F1){$b$}
\drawedge(A1,E1){$a$}
\drawedge[ELside=r](D1,F1){$b$}
\drawedge[curvedepth=3](E1,B1){$b$}
\end{picture}
\end{center}
\caption{Automata ${\mathrsfs A}_{6}$ and ${\mathrsfs B}_{6}$.}
\label{fig Z6}
\end{figure}

The transition function $\delta$ of the first DFA $\mathrsfs{A}_{n+1}$ is defined as follows:
\begin{center}
$\delta(i,a)=\begin{cases}
    i+1 & \text{ if } 0< i< n,\\
    0 & \text{ if } i=0,\\
    n & \text{ if } i=n,
  \end{cases}
\qquad\delta(i,b)=\begin{cases}
    0 & \text{ if } 0< i< n,\\
    1 & \text{ if } i=0 \text{ or } i=n.
  \end{cases}$
\end{center}
Note that $\mathrsfs{A}_{n+1}$ is strongly connected. Indeed, the states $1,2,\ldots, n$ appear in a cycle marked by the word $a^{n-1}b$, furthermore  $0\dt b=1$ and $1\dt b=0$. Now we need to check, that the language of words synchronizing the automaton $\mathrsfs{A}_{n+1}$ coincides with $\Sigma^*a^{n-1}b\Sigma^*$. The standard tool for finding the language of synchronizing words of a given DFA 
$\mathrsfs{A}=\langle Q,\delta,\Sigma\rangle$ is the \emph{power automaton} $\mathcal{P}(\mathrsfs{A})$. Its state set is the set $\mathcal{Q}$ of all nonempty subsets of $Q$, and the transition function is defined as a natural extension of $\delta$ on the set $\mathcal{Q}\times\Sigma$ (the resulting function is also denoted by $\delta$), namely $\delta(S,a)=\{\delta(s,a)\mid s\in S\}$ for $S\subseteq Q$ and $a\in\Sigma$. 
The automaton $\mathcal{P}(\mathrsfs{A})$ recognizes $\Syn(\mathrsfs{A})$ provided one takes $Q$ as the initial state and singletons as final states. 
It is easy to see, that if all the singletons are identified to obtain unique sink state $s$, the resulting automaton still recognizes $\Syn(\mathrsfs{A})$. Throughout the paper the term \emph{power automaton} will refer to this modified version. The Fig.~\ref{figA1} shows the power automaton for the language $\mathrsfs{A}_{n+1}$ (for clarity only reachable from $Q$ subsets are shown). From the structure of $\mathcal{P}(\mathrsfs{A}_{n+1})$ it is easy to see that the language of synchronizing words of the automaton $\mathrsfs{A}_{n+1}$ coincides with $L$.

Next we consider the DFA $\mathscr{B}_{n+1}$ with $n+1=2k$ and transition function $\delta$ defined by the rule
\begin{center}
$\delta(i,a)=\begin{cases}
    i+1 & \text{ if } 0\leq i< n,\\
    n-1 & \text{ if } i=n,
  \end{cases}
\qquad\delta(i,b)=\begin{cases}
    0 & \text{ if } i \text{ is odd and } i\neq n,\\
    1 & \text{ if } i \text{ is even or } i=n.
  \end{cases}$
\end{center}
We verify that $\mathscr{B}_{n+1}$ is strongly connected. Indeed, the states 1,2,..., $n$ appear in a cycle marked by the word $a^{n-1}b$, furthermore $0\dt a=1$ and $1\dt b=0$. It is easily seen that, for any odd $n$, $\mathscr{B}_{n+1}$ and $\mathscr{A}_{n+1}$ are not isomorphic. For the power automaton $\mathcal{P}(\mathscr{B}_{2k})$ see the left side of Fig.~\ref{figB1} (again, only reachable from $Q$ subsets are shown).
\begin{figure}[ht]
\begin{center}
\unitlength=2pt
\begin{picture}(60,75)
\gasset{Nw=10,Nh=10,Nmr=5}
\thinlines
\node[Nadjust=wh,Nmr=1](A)(0,75){$0,1,2,...,n$}
\node[Nadjust=wh,Nmr=1](B)(0,60){$0,2,3,...,n$}
\node[Nframe=n](p0)(0,45){$\ldots$}
\node[Nadjust=wh,Nmr=1](D)(0,30){$0,n-1,n$}
\node[Nadjust=wh,Nmr=1](E)(0,15){$0,n$}
\node(f)(0,0){$s$}
\node[Nadjust=wh,Nmr=1](B1)(50,60){$0,1$}
\node[Nadjust=wh,Nmr=1](C1)(50,45){$0,2$}
\node[Nadjust=wh,Nmr=1](D1)(50,30){$0,3$}
\node[Nframe=n](p1)(50,15){$\ldots$}
\node[Nadjust=wh,Nmr=1](E1)(50,0){$0,n-1$}
\drawloop[loopdiam=8,loopangle=45](B1){$b$}
\drawloop[loopdiam=8,loopangle=180](E){$a$}
\drawloop[loopdiam=8,loopangle=180](f){$a,b$}
\drawedge(A,B){$a$}
\drawedge(B,p0){$a$}
\drawedge(p0,D){$a$}
\drawedge(D,E){$a$}
\drawedge(A,B1){$b$}
\drawedge(B,B1){$b$}
\drawedge(D,B1){$b$}
\drawedge(E,f){$b$}
\drawedge(B1,C1){$a$}
\drawedge(C1,D1){$a$}
\drawedge(D1,p1){$a$}
\drawedge(p1,E1){$a$}
\drawedge(E1,E){$a$}
\drawedge[curvedepth=-6,ELside=r](C1,B1){$b$}
\drawedge[curvedepth=-16,ELside=r](D1,B1){$b$}
\drawedge[curvedepth=10](E1,B1){$b$}
\end{picture}
\end{center}
\caption{The power automaton $\mathcal{P}(\mathscr{A}_{n+1})$}
\label{figA1}
\end{figure}
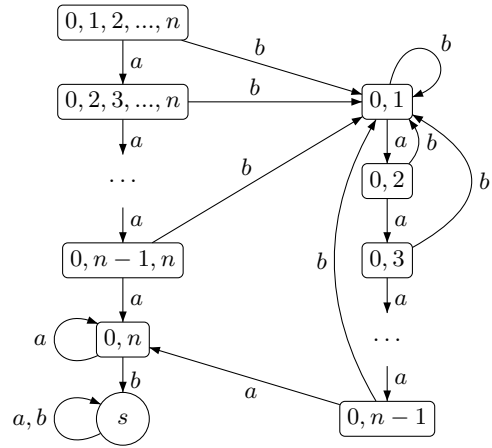
It remains to construct a series of strongly connected automata $\mathscr{B}_{n+1}$ with $n+1=2k+1$. The transition function $\delta$ is defined as follows:
\begin{center}
$\delta(i,a)=\begin{cases}
    i+1 & \text{ if } 1< i< n \text{ or } i=0,\\
    0 & \text{ if } i=1 \text{ or } i=n,
  \end{cases}
\qquad\delta(i,b)=\begin{cases}
    0 & \text{ if } 1< i \leq n,\\
    2 & \text{ if } i=0 \text{ or } i=1.
  \end{cases}$
\end{center}
\begin{figure}[ht]
\begin{center}
\unitlength=2pt
\begin{picture}(125,95)
\gasset{Nw=10,Nh=10,Nmr=5}
\thinlines
\node[Nadjust=wh,Nmr=1](A)(0,90){$0,1,2,...,n$}
\node[Nadjust=wh,Nmr=1](B)(0,75){$1,2,3,...,n$}
\node[Nadjust=wh,Nmr=1](C)(0,60){$2,3,4,...,n$}
\node[Nframe=n](p0)(0,45){$\ldots$}
\node[Nadjust=wh,Nmr=1](D)(0,30){$n-2,n-1,n$}
\node[Nadjust=wh,Nmr=1](E)(0,15){$n-1,n$}
\node(f)(0,0){$s$}
\node[Nadjust=wh,Nmr=1](B1)(40,80){$0,1$}
\node[Nadjust=wh,Nmr=1](C1)(40,60){$1,2$}
\node[Nadjust=wh,Nmr=1](D1)(40,45){$2,3$}
\node[Nframe=n](p1)(40,30){$\ldots$}
\node[Nadjust=wh,Nmr=1](E1)(40,15){$n-2,n-1$}
\drawloop[loopdiam=8,loopangle=45](B1){$b$}
\drawloop[loopdiam=8,loopangle=180](E){$a$}
\drawloop[loopdiam=8,loopangle=180](f){$a,b$}
\drawedge(A,B){$a$}
\drawedge(B,C){$a$}
\drawedge(C,p0){$a$}
\drawedge(p0,D){$a$}
\drawedge(D,E){$a$}
\drawedge(A,B1){$b$}
\drawedge(B,B1){$b$}
\drawedge(C,B1){$b$}
\drawedge(D,B1){$b$}
\drawedge(E,f){$b$}
\drawedge(B1,C1){$a$}
\drawedge(C1,D1){$a$}   
\drawedge(D1,p1){$a$}
\drawedge(p1,E1){$a$}
\drawedge(E1,E){$a$}
\drawedge[curvedepth=-8,ELside=r](C1,B1){$b$}
\drawedge[curvedepth=-16,ELside=r](D1,B1){$b$}
\drawedge[curvedepth=10](E1,B1){$b$}
\node[Nadjust=wh,Nmr=1](A2)(82,90){$0,1,2,...,n$}
\node[Nadjust=wh,Nmr=1](B2)(82,75){$0,1,3,...,n$}
\node[Nadjust=wh,Nmr=1](C2)(82,60){$0,1,4,...,n$}
\node[Nframe=n](p02)(82,45){$\ldots$}
\node[Nadjust=wh,Nmr=1](D2)(82,30){$0,1,n$}
\node[Nadjust=wh,Nmr=1](E2)(82,15){$0,1$}
\node(f2)(82,0){$s$}
\node[Nadjust=wh,Nmr=1](B12)(122,80){$0,2$}
\node[Nadjust=wh,Nmr=1](C12)(122,60){$1,3$}
\node[Nadjust=wh,Nmr=1](D12)(122,45){$0,4$}
\node[Nframe=n](p12)(125,30){$\ldots$}
\node[Nadjust=wh,Nmr=1](E12)(125,15){$0,n$}
\drawloop[loopdiam=8,loopangle=45](B12){$b$}
\drawloop[loopdiam=8,loopangle=180](E2){$a$}
\drawloop[loopdiam=8,loopangle=180](f2){$a,b$}
\drawedge(A2,B2){$a$}
\drawedge(B2,C2){$a$}
\drawedge(C2,p02){$a$}
\drawedge(p02,D2){$a$}
\drawedge(D2,E2){$a$}
\drawedge(A2,B12){$b$}
\drawedge(B2,B12){$b$}
\drawedge(C2,B12){$b$}
\drawedge(D2,B12){$b$}
\drawedge(E2,f2){$b$}
\drawedge(B12,C12){$a$}
\drawedge(C12,D12){$a$}
\drawedge(D12,p12){$a$}
\drawedge(p12,E12){$a$}
\drawedge(E12,E2){$a$}
\drawedge[curvedepth=-8,ELside=r](C12,B12){$b$}
\drawedge[curvedepth=-16,ELside=r](D12,B12){$b$}
\drawedge[curvedepth=10](E12,B12){$b$}
\end{picture}
\end{center}
\caption{The power automata $\mathcal{P}$($\mathscr{B}_{2k}$) and $\mathcal{P}$($\mathscr{B}_{2k+1}$)}
\label{figB1}
\end{figure}
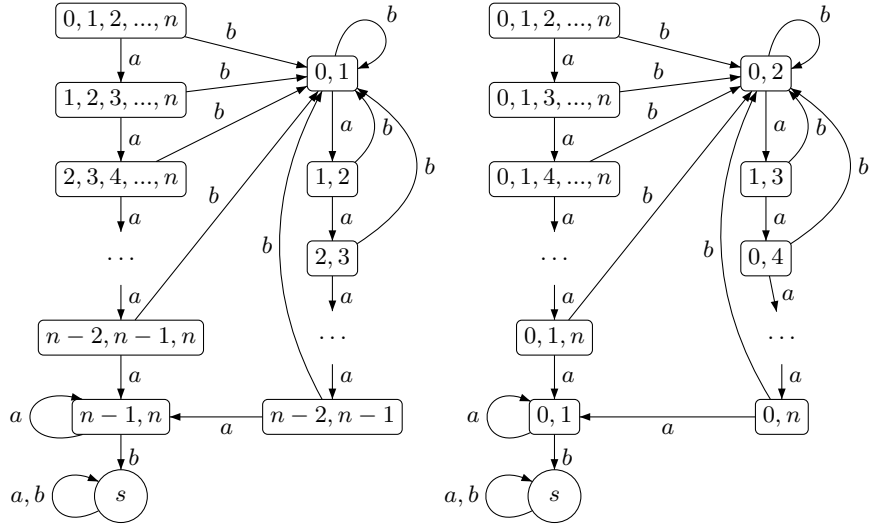
The states 2,3,...,$n$, 0 form a cycle marked by the word $a^{n-1}b$, furthermore $0\dt a=1$ and $1\dt a=0$. It is easily seen that, for any even $n$, $\mathscr{B}_{n+1}$ and $\mathscr{A}_{n+1}$ are not isomorphic. The power automaton $\mathcal{P}(\mathrsfs{B}_{2k+1})$ consisting only of reachable from $Q$ subsets is on the right side of Fig.~\ref{figB1}. From the structure of the power automaton for $\mathscr{B}_{n+1}$ it easily follows that its set of synchronizing words coincides with $L$.

\section{Algorithm}

\subsection{Formal description}

In this section we provide an algorithm to construct a strongly connected synchronizing automaton whose language of synchronizing words is generated by a single word $w$. The minimal automaton recognizing this language is denoted by $\mathrsfs{A}_w$. The main idea of the construction is the following. We try to construct a strongly connected automaton such that in its \emph{pair automaton} there is a subautomaton isomorphic to $\mathrsfs{A}_w$. 
Our algorithm can be applied in case of an arbitrary alphabet, but for clarity we explain it only in binary case.

Recall, that the pair automaton of a given DFA $\mathrsfs{A}=\langle Q,\Sigma,\delta\rangle$ is the subautomaton $\mathcal{P}^{[2]}(\mathrsfs{A})$ of the power automaton $\mathcal{P}(\mathrsfs{A})$ consisting only of 2-element subsets of $Q$ and the sink state $s$.

 Fix a word $w$ over $\Sigma=\{a,b\}$. Let $|w|=n$. Without loss of generality suppose that the first letter of $w$ is $a$. Denote the $i$-th letter of $w$ by $w[i]$ and the prefix $w[1]w[2]...w[i]$ by $w[1..i]$. For any letter $x\in\{a,b\}$ by $\overline{x}$ we denote its \emph{complementary} letter, i.e. $\overline{a}=b$, and $\overline{b}=a$.
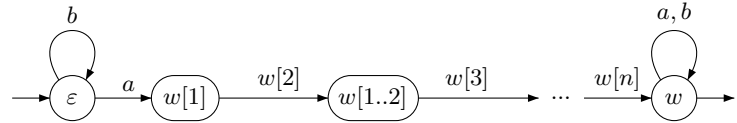
\begin{figure}[ht]
\begin{center}
  \begin{picture}(95,7)
   \gasset{Nw=6,Nh=6,Nmr=3}
   \thinlines
   \node[Nmarks=i](A1)(10,0){$\varepsilon$}
   \node[Nw=9](A2)(25,0){$w[1]$}
   \node[Nw=12](A3)(50,0){$w[1..2]$}
   \node[Nmarks=f](A4)(90,0){$w$}
   \node[Nframe=n](A5)(75,0){$...$}
   \drawloop[loopdiam=6,loopangle=90](A4){$a,b$}
   \drawloop[loopdiam=6,loopangle=90](A1){$b$}
   \drawedge(A1,A2){$a$}
   \drawedge(A2,A3){$w[2]$}
   \drawedge(A5,A4){$w[n]$}
   \drawedge(A3,A5){$w[3]$}
   \end{picture}
\end{center}
\caption{The minimal DFA $\mathscr{A}_w$.}
\label{min_aut}
\end{figure}
Let us remind the construction of the minimal automaton recognizing the language $\Sigma^{*}w\Sigma^{*}$. It is well-known that this automaton has $n+1$ states. We enumerate the states of this automaton by the prefixes of the word $w$ so that the state $w[1..i]$ maps to the state $w[1..i+1]$ under the action of the letter $w[i+1]$ for all $i$, $0\leq i<n$. The other letter $\overline{w[i+1]}$ maps the state $w[1..i]$ to the state $p$ such that $p$ is the maximal prefix of $w$ that appears in the word $w[1..i+1]$ as a suffix. The state $w$ is the sink state.
The initial state is $\varepsilon$ and the unique final state is $w$, see Fig.\ref{min_aut} (the transitions labeled by complementary letters $w[i]$ are not shown).

The algorithm constructing a required strongly connected synchronizing automaton $\mathrsfs{B}$ with the state set  $Q=\{0,1,\ldots,n\}$ proceeds inductively. On the first step we put $Q=\{0,1,2\}$ and define the action of letters on the states $0$ and $1$. On the $i^\text{th}$ step ($1<i<n$) 
we add new state $i+1$ to $Q$ and define the transition function on the state $i$. On the last, $n^\text{th}$ step we define the transition function on the state $n$.
 Transitions on each step are defined in such a way, that after the $i^\text{th}$ step of the algorithm ($1\le i \le n$)  the current pair automaton has a subautomaton isomorphic to the part of the minimal automaton $\mathrsfs{A}_w$ consisting of the states $\varepsilon, w[1],\ldots, w[1..i]$.

Consider the first step of the algorithm.
We need to associate the states $\varepsilon$ and $w[1]$ of $\mathscr{A}_w$ with 2-element subsets $\{p_i,q_i\}$ of $Q=\{0,1,2\}$. Without loss of generality we associate the state $\varepsilon$ with the subset $\{0,1\}$, and the state $w[1]$ with the subset $\{1,2\}$. In the automaton $\mathscr{A}_w$ the state $\varepsilon$ is fixed by $b$ and maps to $w[1]$ under the action of $a$. Define in $\mathscr{B}$ the transition function on $0$ and $1$ in such a way that the subset $\{0,1\}$ in $\mathscr{B}$ is fixed under the action of $b$, and maps to the subset $\{1,2\}$ under the action of  $a$. We have four different ways to do so (see Fig.\ref{buildB}).
\begin{figure}[ht]
\begin{center}
  \begin{picture}(105,20)
   \gasset{Nw=6,Nh=6,Nmr=3}
   \thinlines
   \node(A1)(0,15){$0$}
   \node(A2)(15,15){$1$}
   \node(A3)(15,0){$2$}
   \drawloop[loopdiam=6](A2){$a$}
   \drawedge[curvedepth=3](A1,A2){$b$}
   \drawedge(A2,A1){$b$}
   \drawedge(A1,A3){$a$}
   \node(B1)(30,15){$0$}
  \node(B2)(45,15){$1$}
   \node(B3)(45,0){$2$}
   \drawloop[loopdiam=6,loopangle=90](B1){$b$}
   \drawloop[loopdiam=6,loopangle=90](B2){$a,b$}
   \drawedge(B1,B3){$a$}
   \node(C1)(60,15){$0$}
   \node(C2)(75,15){$1$}
   \node(C3)(75,0){$2$}
   \drawedge(C2,C3){$a$}
   \drawloop[loopdiam=6](C1){$b$}
   \drawloop[loopdiam=6](C2){$b$}
   \drawedge(C1,C2){$a$}
   \node(D1)(90,15){$0$}
   \node(D2)(105,15){$1$}
   \node(D3)(105,0){$2$}
   \drawedge[curvedepth=3](D1,D2){$a,b$}
   \drawedge(D2,D1){$b$}
   \drawedge(D2,D3){$a$}
   \end{picture}
\end{center}
\caption{Possible transitions from 0 and 1 in the automaton $\mathscr{B}$.}
\label{buildB}
\end{figure}
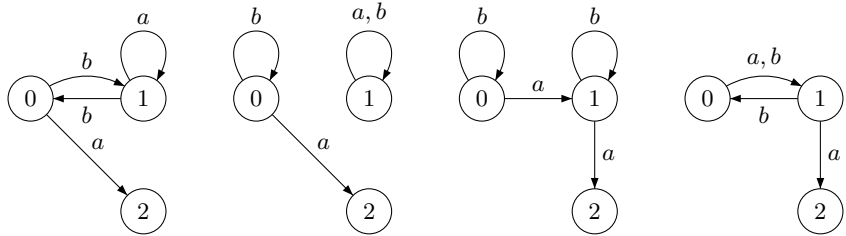
 It is easy to see that in fact the second case is impossible since the DFA $\mathscr{B}$ will not be strongly connected, not even during the rest of the construction. For certainty consider the first variant, so $0\dt a=2$, $0\dt b=1$, $1\dt a=1$, and $1\dt b=0$.
%

Let us describe the $i^\text{th}$ step of the algorithm. We have $Q=\{0,1,\ldots, i\}$. For convenience we represent the subsets $\{p_i,q_i\}\ne\{0,1\}$ of $Q$ as ordered pairs $(p_i,q_i)$ with $p_i>q_i$, and the subset $\{0,1\}$ as pair $(0,1)$. On previous steps the states $\varepsilon,w[1],\ldots,w[1..i-1]$ of $\mathscr{A}_w$ were associated with the pairs $(p_0,q_0)$, $(p_1,q_1)$, \ldots, $(p_{i-1},q_{i-1})$ in $\mathcal{P}^{[2]}(\mathrsfs{B})$ in such a way that $p_0=0,p_1=2,p_2=3$,$\ldots, p_{i-1}=i$ and the transition function on the states $0,1,...,i-1$ was defined (see the dash-dotted part on Fig.~\ref{step_i}). We add the state $i+1$ to $Q$. Next we associate the state $w[1..i]$ of $\mathscr{A}_w$ with the pair $(i+1,q_i)$, where $q_i=q_{i-1}\dt w[i]$, and put $i\dt w[i]=i+1$. 
It remains to define the transition $i\dt \overline{w[i]}$. Let $w[1..j]=w[1..i-1]\dt \overline{w[i]}$, and let $(p_j,q_j)$ be the associated pair of states in $\mathcal{P}^{[2]}(\mathrsfs{B})$. In the correctness section we will show that one of the equalities holds: either $q_{i-1}\dt \overline{w[i]}=q_j$ or $q_{i-1}\dt\overline{w[i]}=p_j$. If $q_{i-1}\dt \overline{w[i]}=q_j$, then we put $i\dt \overline{w[i]}=p_j$, otherwise, $i\dt \overline{w[i]}=q_j$. The $i^\text{th}$ step is illustrated on Fig.~\ref{step_i}. The left part of the picture corresponds to the current pair automaton (only states currently associated to the states of $\mathscr{A}_w$ are shown), the right part is the corresponding subautomaton of $\mathscr{A}_w$. The correspondence between states of $\mathscr{A}_w$ and those of pair automaton is shown in dashed lines. The transitions defined on this step are shown in thick lines.
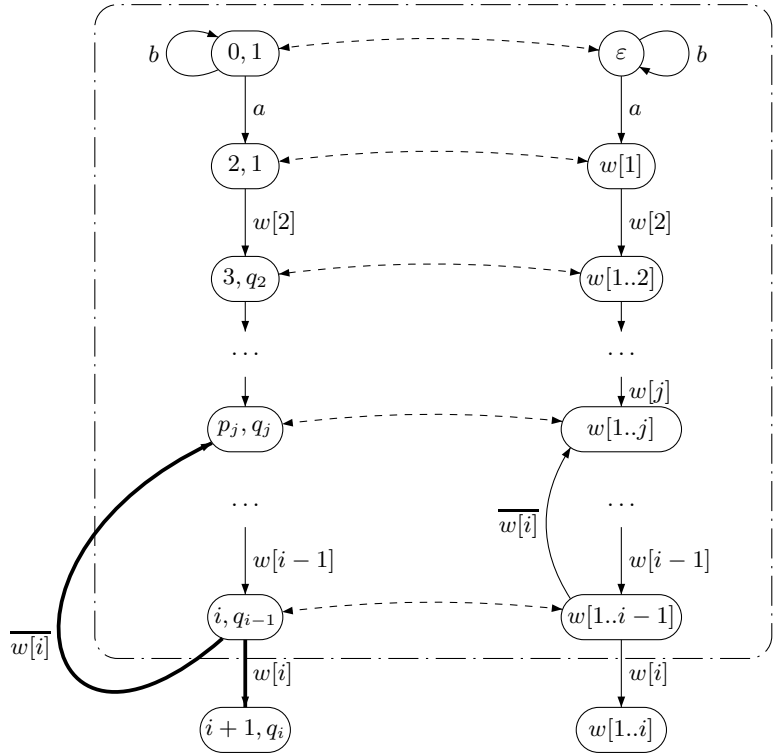
\begin{figure}[ht]
\begin{center}
  \begin{picture}(50,95)
   \gasset{Nw=6,Nh=6,Nmr=3}
   \thinlines
   \node(A1)(50,90){$\varepsilon$}
   \node[Nw=9](A2)(50,75){$w[1]$}
   \node[Nw=11](A3)(50,60){$w[1..2]$}
   \node[Nframe=n](A4)(50,50){$\ldots$}
   \node[Nw=16](A5)(50,40){$w[1..j]$}
   \node[Nframe=n](A6)(50,30){$\ldots$}
   \node[Nw=16](A7)(50,15){$w[1..i-1]$}
   \node[Nw=12](A8)(50,0){$w[1..i]$}
   \drawloop[loopdiam=6,loopangle=0](A1){$b$}
   \drawedge(A1,A2){$a$}
   \drawedge(A2,A3){$w[2]$}
   \drawedge(A3,A4){}
   \drawedge(A4,A5){$w[j]$}
   \drawedge(A6,A7){$w[i-1]$}
   \drawedge(A7,A8){$w[i]$}
   \drawedge[curvedepth=5,sxo=-5,exo=-5](A7,A5){$\overline{w[i]}$}
   \node[Nw=9](B1)(0,90){$0,1$}
   \node[Nw=9](B2)(0,75){$2,1$}
   \node[Nw=9](B3)(0,60){$3,q_2$}
   \node[Nframe=n](B4)(0,50){$\ldots$}
   \node[Nw=10](B5)(0,40){$p_j,q_j$}
   \node[Nframe=n](B6)(0,30){$\ldots$}
   \node[Nw=10](B7)(0,15){$i,q_{i-1}$}
   \node[Nw=12](B8)(0,0){$i+1,q_i$}
   \drawloop[loopdiam=6,loopangle=180](B1){$b$}
   \drawedge(B1,B2){$a$}
   \drawedge(B2,B3){$w[2]$}
   \drawedge(B3,B4){}
   \drawedge(B4,B5){}
   \drawedge(B6,B7){$w[i-1]$}
   \drawedge[linewidth=0.5](B7,B8){$w[i]$}
   \drawbpedge[linewidth=0.5](B7,225,40,B5,200,40){$\overline{w[i]}$}
   \drawedge[ATnb=1,curvedepth=-2,dash={1}0](A1,B1){}
   \drawedge[ATnb=1,curvedepth=-2,dash={1}0](A2,B2){}
   \drawedge[ATnb=1,curvedepth=-2,dash={1}0](A3,B3){}
   \drawedge[ATnb=1,curvedepth=-2,dash={1}0](A5,B5){}
   \drawedge[ATnb=1,curvedepth=-2,dash={1}0](A7,B7){}
   \node[dash={4 1 0.3 1}0, Nw=90,Nh=87](C)(25,53){}
   \end{picture}
\end{center}
\caption{Step $i$: Current pair automaton, and the corresponding part of $\mathrsfs{A}_w$.}
\label{step_i}
\end{figure}

On the last $n^{\text{th}}$ step we map the state $n$ of $\mathrsfs{B}$ to the state $q_{n-1}\dt w[n]$ under the action of the letter $w[n]$ in order to associate the state $w$ of $\mathrsfs{A}_w$ with the sink state $s$ of $\mathcal{P}^{[2]}(\mathrsfs{B})$. The action of the letter $\overline{w[n]}$ is defined as before.

Obviously the language consisting of words synchronizing the subset $\{0,1\}$ coincides with $\Sigma^*w\Sigma^*$. Since the set of words synchronizing the whole automaton $\mathrsfs{B}$ is contained in the set of words synchronizing any of its subsets, we have $\Syn(\mathrsfs{B})\subseteq\Sigma^*w\Sigma^*$. Let
us show that the word $w$ synchronizes $\mathrsfs{B}$. Let $0\dt w=1\dt w=m\in Q$. We proceed inductively. 
Assume that for any $0\le k< i$ we have $k\dt w=m$. Let us show that $i\dt w=m$. The state $i$ belongs to the pair $(i,q)$ in the subautomaton of $\mathcal{P}^{[2]}(\mathrsfs{B})$ isomorphic to $\mathrsfs{A}_w$. Thus, the pair $(i,q)$
is synchronized by the word $w$. Since $q<i$, by induction hypothesis we have $q\dt w=m$, therefore $i\dt w=m.$
Finally, we have $n\dt w=m$, so $w\in\Syn(\mathrsfs{B})$. Hence $\Sigma^*w\Sigma^*\subseteq\Syn(\mathrsfs{B})$.


\begin{example}
We apply our algorithm to the word $aabab$. First, build the minimal automaton recognizing $\Sigma^{*}aabab\Sigma^{*}$ (see Fig.~\ref{aabab}).
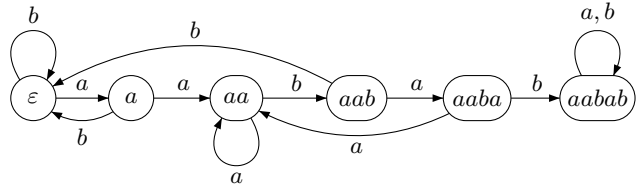
\begin{figure}[ht]
\begin{center}
\begin{picture}(75,12)
   \gasset{Nw=6,Nh=6,Nmr=3}
\thinlines
\node(A)(0,5){$\varepsilon$}
\node(B)(13,5){$a$}
\node[Nw=7](C)(27,5){$aa$}
\node[Nw=8](D)(43,5){$aab$}
\node[Nw=9](E)(59,5){$aaba$}
\node[Nw=10](F)(75,5){$aabab$}
\drawedge(A,B){$a$}
\drawedge(B,C){$a$}
\drawedge(C,D){$b$}
\drawedge(D,E){$a$}
\drawedge(E,F){$b$}
\drawloop[loopdiam=6,loopangle=90](A){$b$}
\drawloop[loopdiam=6,loopangle=-90](C){$a$}
\drawloop[loopdiam=6,loopangle=90](F){$a,b$}
\drawedge[curvedepth=3](B,A){$b$}
\drawedge[curvedepth=-7,ELside=r](D,A){$b$}
\drawedge[curvedepth=5](E,C){$a$}
\end{picture}
\end{center}
\caption{The minimal DFA recognizing $\Sigma^{*}aabab\Sigma^{*}$.}
\label{aabab}
\end{figure}
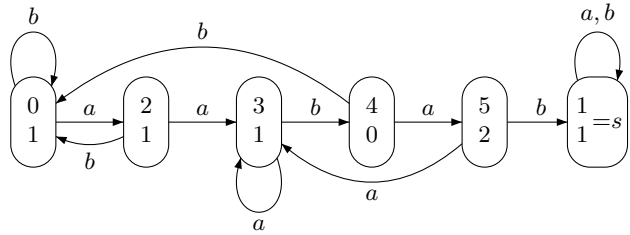
\begin{figure}[ht]
\begin{center}
\begin{picture}(75,22)
   \gasset{Nw=6,Nh=12,Nmr=3}
\thinlines
\node(A)(0,10){$\begin{array}{c}
                   0 \\
                   1
                 \end{array}$}
\node(B)(15,10){$\begin{array}{c}
                   2 \\
                   1
                 \end{array}$}
\node(C)(30,10){$\begin{array}{c}
                   3 \\
                   1
                 \end{array}$}
\node(D)(45,10){$\begin{array}{c}
                   4 \\
                   0
                 \end{array}$}
\node(E)(60,10){$\begin{array}{c}
                   5 \\
                   2
                 \end{array}$}
\node[Nw=8](F)(75,10){$\begin{array}{c}
                   1 \\
                   1
                 \end{array}$=$s$}
\drawedge(A,B){$a$}
\drawedge(B,C){$a$}
\drawedge(C,D){$b$}
\drawedge(D,E){$a$}
\drawedge(E,F){$b$}
\drawloop[loopdiam=6,loopangle=90](A){$b$}
\drawloop[loopdiam=6,loopangle=-90](C){$a$}
\drawloop[loopdiam=6,loopangle=90](F){$a,b$}
\drawedge[curvedepth=3](B,A){$b$}
\drawedge[curvedepth=-10,ELside=r](D,A){$b$}
\drawedge[curvedepth=8](E,C){$a$}
\end{picture}
\end{center}
\caption{The corresponding subautomaton in the pair automaton $\mathcal{P}^{[2]}(\mathrsfs{B})$}
\label{aabab_pair}
\end{figure}
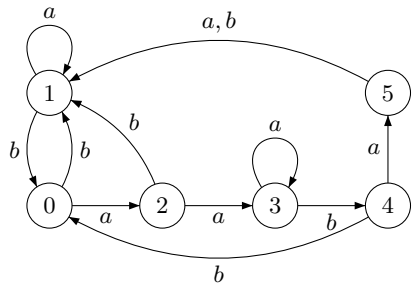
\begin{figure}[ht]
\begin{center}
\begin{picture}(55,30)
   \gasset{Nw=6,Nh=6,Nmr=3}
\thinlines
\node(A)(0,5){$0$}
\node(B)(0,20){$1$}
\node(C)(15,5){$2$}
\node(D)(30,5){$3$}
\node(E)(45,5){$4$}
\node(F)(45,20){$5$}
\drawedge[curvedepth=-3,ELside=r](A,B){$b$}
\drawedge[curvedepth=-3,ELside=r](B,A){$b$}
\drawedge[ELside=r](A,C){$a$}
\drawedge[ELside=r](C,D){$a$}
\drawedge[ELside=r](D,E){$b$}
\drawedge[curvedepth=-7,ELside=r](F,B){$a,b$}
\drawedge[curvedepth=-3,ELside=r](C,B){$b$}
\drawedge(E,F){$a$}
\drawloop[loopdiam=6,loopangle=90](D){$a$}
\drawloop[loopdiam=6,loopangle=90](B){$a$}
\drawedge[curvedepth=7](E,A){$b$}
\end{picture}
\end{center}
\caption{Strongly connected synchronizing automaton $\mathscr{B}$ with $Syn(\mathscr{B})=\Sigma^{*}aabab\Sigma^{*}$.}
\label{aababB}
\end{figure}
Next we show in details the construction of the DFA $\mathscr{B}$.

\noindent \emph{Step 1.} We have $Q=\{0,1,2\}$, and $0\dt a=2$, $0\dt b=1$, $1\dt a=1$, and $1\dt b=0$.

\noindent \emph{Step 2.} Add state $3$ to $Q$ and put $2\dt a=3$. Associate
the state $aa$ of $\mathrsfs{A}_w$ with the pair $(3,q_2)$, where $q_2=1\dt a=1$. To define $2\dt b$ we see that the condition $\{1,2\}\dt b=\{0,1\}$
must be satisfied. Since $1\dt b=0$, we put $2\dt b=1$.

\noindent \emph{Step 3}. Add state $4$ to $Q$ and put $3\dt b=4$. Associate
the state $aab$ of $\mathrsfs{A}_w$ with the pair $(4,q_3)$, where $q_3=1\dt b=0$. 
To define $3\dt a$ we see that the condition $\{1,3\}\dt a=\{1,3\}$
must be satisfied. Since $1\dt a=1$, we put $3\dt a=3$.

\noindent \emph{Step 4}. Add state $5$ to $Q$ and put $4\dt a=5$. Associate
the state $aaba$ of $\mathrsfs{A}_w$ with the pair $(5,q_4)$, where $q_4=0\dt a=2$. 
To define $4\dt b$ we see that the condition $\{0,4\}\dt b=\{0,1\}$
must be satisfied. Since $0\dt b=1$, we put $4\dt b=0$.

\noindent \emph{Step 5}. This is the last step of the algorithm. We do not add any new states,
only define the transition function on the state $5$. We put $5\dt b=2\dt b=1$
in order to have the pair $(5,2)$ associated with the sink state of the pair automaton.
Since the condition $\{2,5\}\dt a=\{1,3\}$
must be satisfied, and $2\dt a=3$, we put $5\dt a=1$.
The resulting pairs associated with the states of $\mathscr{A}_w$ are shown on Fig.\ref{aabab_pair}.
The corresponding strongly connected DFA $\mathscr{B}$ is shown on Fig.~\ref{aababB}.
\end{example}

\subsection{Correctness}

Now we prove the correctness of the algorithm. The proof consists of two stages. First we verify that there will be no conflict while defining the action of letters on the states of $\mathrsfs{B}$. Next we show that the resulting automaton $\mathrsfs{B}$ is strongly connected.

\textbf{Stage 1.} Consider the word $a^n$. Construct the DFA $\mathrsfs{B}$ using the algorithm. Its transition function is defined as follows:
\begin{center}
$\delta_\mathrsfs{B}(i,a)=\begin{cases}
    i+1 & \text{ if } 1< i< n;\\
    1 & \text{ if } i=1 \text{ or } i=n;\\
    2 & \text{ if } i=0.
  \end{cases} \qquad \delta_\mathrsfs{B}(i,b)=\begin{cases}
    1 & \text{if } i=0 \text{ or }1<i\leq n;\\
    0 & \text{if } i=1.
  \end{cases}$
\end{center}
Obviously, the automaton is strongly connected and all assignments are correct. Thus we assume that the word $w$ contains $b$, i.e. $w=a^{k}bv$, $v\in\Sigma^*$. We need to show that on $i^\text{th}$ step of the algorithm either $q_{i-1}\dt \overline{w[i]}=p_j$ or $q_{i-1}\dt \overline{w[i]}=q_j$, where
$(p_j,q_j)$ is the pair associated with the prefix $w[1..j]$ of $w$.
Since $w=a^{k}bv$,  the states $w[1],w[1..2],\ldots,w[1..k-1]$ in $\mathscr{A}_w$ are mapped under the action of $b$ to $\varepsilon$. By the construction of $\mathrsfs{B}$ the state $\varepsilon$ is associated with the pair $(0,1)$. Thus,
the letter $b$ maps the states $p_0,p_1,\ldots,p_{k-1}$ to 0 or 1.  
Since the state $1$ is fixed by the letter $a$, we have that $q_0=1$, $q_1=1$,...,$q_k=1$. Since the letter $b$ maps the state $1$ to $0$, we have $q_{k+1}=0$. Note that $p_{\ell-1}$ is equal to the length of the prefix $w[1..\ell]$ of $w$. If $v$ and $w$ have common non-empty prefix, then $q_{k+2}=2$, $q_{k+3}=3$, etc. So the sequence of $q$'s consists of blocks of 1's and possibly blocks $0,2,3...$ corresponding to the length of some prefix. Let $u$ be the maximal suffix of the word $w[1..i-1]$ which appears in this word as a prefix. Then in the pair $(p_{i-1},q_{i-1})$ the state $q_{i-1}$ is either $p_{|u|}$ or $q_{|u|}$. It is clear that in $\mathscr{A}_w$ the state $|u|$ under the action of $\overline{w[i]}$ maps to $j$. Thus we can define the action of $\overline{w[i]}$ on the state $p_{i-1}$ to make the pair $(p_{i-1},q_{i-1})$ map to $(p_j,q_j)$. Indeed, if $q_{i-1}=p_{|u|}$ then $p_{i-1}\dt \overline{w[i]}=q_{|u|}\dt\overline{w[i]}$, otherwise put $p_{i-1}\dt \overline{w[i]}=p_{|u|}\dt\overline{w[i]}$.

\textbf{Stage 2.} Now we prove that the resulting automaton $\mathrsfs{B}$ is strongly connected.
Denote by $\mathrsfs{A}'_w$ the minimal automaton $\mathrsfs{A}_w$ after deleting the sink state. Let us assume first that $\mathrsfs{A}'_w$ is strongly connected. Let $i$ be an arbitrary state in the automaton $\mathrsfs{B}$. By assumption there is a path in $\mathcal{P}^{[2]}(\mathrsfs{B})$
from the pair $(i,q_{i-1})$ (associated with some state in $\mathrsfs{A}_w$) to the pair $(0,1)$ (associated with the state $\varepsilon$). Hence, either $0$ or $1$ is reachable from the state $i$. In fact both states $0$
and $1$ are reachable from $i$, since $0\dt b=1$ and $1\dt b=0$. By the construction the pair $(i,q_{i-1})$ is reachable from the pair
$(0,1)$, so the state $i$ is reachable both from $0$ and $1$ in $\mathrsfs{B}$. Thus, in this case the automaton $\mathrsfs{B}$ is strongly connected.

 Let us prove that if $w\notin\{a^{n-1}b, ab^{n-1}\}$, then the automaton $\mathrsfs{A}'_w$ is strongly connected. Consider an arbitrary state $w[1..i]$
 of $\mathrsfs{A}'_w$. This state is obviously reachable from the state $\varepsilon$. We show that the state $\varepsilon$ is reachable from $w[1..i]$.
 Let $w[1..j]=w[1..i]\dt c^{|w|}$, where $c=\overline{w[i+1]}$. The state $w[1..j]$ is a maximal prefix of $w$ of the form $c^k$. If $c=b$,
 then $w[1..j]=\varepsilon$, so we are done. Suppose $c=a$. Apply $b$ to the state $w[1..j]$ (this is possible, since $w\ne a^{n-1}b$). Next we apply $\overline{w[j+2]}$. If $\overline{w[j+2]}=b$, then $w[1..j]\dt bb =\varepsilon$. If $\overline{w[j+2]}=a$, then
 $w[1..j]\dt bab =\varepsilon$ (since $w\ne ab^{n-1}$). So in both cases the state $\varepsilon$ is reachable from $w[1..i]$.

It remains to apply the algorithm to $w=a^{n-1}b$ and $w=ab^{n-1}$ to make sure that the resulting automata in this case are also strongly connected. Indeed, in case $w=a^{n-1}b$ the transitions of the automaton $\mathscr{B}$ are defined as follows:
\begin{center}
$\delta_\mathscr{B}(i,a)=\begin{cases}
    i+1 & \text{ if } 1< i< n,\\
    i & \text{ if } i=1 \text{ if } i=n,\\
    2 & \text{ if } i=0,
  \end{cases}
\qquad\delta_\mathscr{B}(i,b)=\begin{cases}
    1 & \text{if } i=0 \text{ if }1<i<n,\\
    0 & \text{if } i=1 \text{ if } i=n.
  \end{cases}$
\end{center}
The states $0,2,\ldots,n$ form a cycle marked by the word $a^{n-1}b$, furthermore $0\dt b=1$ and $1\dt b=0$. Thus, $\mathscr{B}$ is strongly connected.

In case $w=ab^{n-1}$ the algorithm constructs the following DFA $\mathscr{B}$:
\begin{center}
$\delta_\mathscr{B}(i,a)=\begin{cases}
    1 & \text{ if } 0\le i\le n \text{ is odd},\\
    2 & \text{ if } 0\le i\le n \text{ is even},
  \end{cases}
\qquad\delta_\mathscr{B}(i,b)=\begin{cases}
    i+1 & \text{if } 0\le i<n, i\ne 1,\\
    0 & \text{if } i=1,
  \end{cases}$
\end{center}
  \begin{center}
$  \delta_\mathscr{B}(n,b)=\begin{cases}
    0 & \text{if } n \text{ is even},\\
    1 & \text{if } n \text{ is odd}.
  \end{cases}$
\end{center}
If $n$ is even, then the states $0,2,\ldots,n$ form a cycle marked by the word $ab^{n-1}$, if $n$ is odd, then then the states $0,2,\ldots,n,1$ form a cycle marked by the word $ab^n$. Thus, $\mathscr{B}$ is strongly connected.

\section{On syntactic semigroup of a principal ideal language}

In the previous section for each word $w$ of length $n$ we constructed a strongly connected synchronizing DFA $\mathrsfs{B}$ with $n+1$ states such that $\Syn(\mathrsfs{B})=\Sigma^*w\Sigma^*$. Is it possible to construct such a DFA with less than $n+1$ states? In \cite{SOFSEM} it was shown that for the case $w\in\{a^n,b^n,a^{n-1}b\}$ we have $rc(L_w)=n+1$, where $L_w=\Sigma^{*}w\Sigma^{*}$ and $|w|=n$. But in general this question remains open. Nevertheless computer experiments show that the answer seems to be negative, i.e. the minimal in terms of reset complexity strongly connected synchronizing DFA has $n+1$ states. Another observation concerns the structure of that DFA. Even for a word $w$ of length 3 there may be several non-isomorphic strongly connected synchronizing automata yielding the minimum of reset complexity. But, as experiments show, the \emph{transition semigroups} of all these automata have the same algebraic structure. In this regard it is interesting to study the structure of the transition semigroup of a synchronizing automaton for which given ideal language serves as the language of synchronizing words.

For an ideal language $L\subseteq \Sigma^{*}$ the \emph{Myhill conguence} \cite{Myhill} $\thickapprox_L$ of $L$ is defined as follows:$$u\thickapprox_L \text{ if and only if }xuy\in L \Leftrightarrow xvy\in L \text{ for all }x,y\in\Sigma^{*}.$$ This congruence is also known as \emph{the syntactic congruence} of L. The quotient semigroup $\Sigma^{+}/\thickapprox_L$ of the relation $\thickapprox_L$ is called the \emph{syntactic semigroup} of L. 

\begin{proposition}
\label{HomImage}
Let $L$ be an ideal language, $S$ the syntactic semigroup of $L$ and $S(\mathrsfs{B})$ the transition semigroup of a synchronizing DFA $\mathrsfs{B}$ for which $L$ serves as the language of synchronizing words. Then $S$ is a homomorphic image of $S(\mathrsfs{B})$.
\end{proposition}

\begin{proof}
Take an arbitrary word $x\in \Sigma^{*}$. Let $[x]$ be the class of $x$ in $S$, and $\{x\}$ the class of $x$ in $S(B)$. Define the map $f:S(B)\rightarrow S$ by the rule $f(\{x\})=[x]$. Check that $f$ is a homomorphism. First, we check the correctness of defining $f$. Consider two words $u$ and $v$ from the same class in $S(B)$, it means that $\{u\}=\{v\}$. Show that in this case $[u]=[v]$. We need to check that for any $x,y\in \Sigma^{*}$ from $xuy\in L$ it follows that $xvy\in L$. By conditions of the proposition $L=Syn(B)$. Since $\{u\}=\{v\}$, then $u$ and $v$ generate the same transformations in the transition semigroup of $B$. But then either both $xuy$ and $xvy$ synchronize $B$, or both do not synchronize. Hence $[u]=[v]$. And finally, $f(\{u\}\{v\})=f(\{uv\})=[uv]=[u][v]=f(\{u\})f(\{v\})$. This completes the proof.
\qed
\end{proof}

Note that this proposition holds for every regular ideal language $L$. It means that if the transition semigroup of any DFA $\mathrsfs{B}$ such that
$\Syn(\mathrsfs{B})=L$ possesses some algebraic property which is preserved under homomorphisms, then also the syntactic semigroup of $L$
must possess this property. Thereby it is interesting to study the structure of the syntactic semigroup of an ideal language.
The \textit{syntactic complexity} $\sigma(L)$ of a regular language $L$ is the cardinality of its syntactic semigroup. The notion of syntactic complexity is studied quite extensively: for surveys of this topic and lists of references we refer the reader to \cite{BrzId,HolK3}. 
In \cite{BrzId} it was conjectured that in case of ideal languages $\sigma(L)\leq n^{n-2}+(n-2)2^{n-2}+1$, where $n$ is the state complexity of $L$. Also in \cite{BrzId} it was shown that there exists an ideal language of syntactic complexity $n^{n-2}+(n-2)2^{n-2}+1$. We consider partial case of principal ideal languages. Recall that $u\in\Sigma^{+}$ is an \textit{inner factor} of $w$ if there exist words $t,s\in\Sigma^{+}$ such that $w=tus$. Denote by $N(w)$ the number of different inner factors of $w$. We prove the following

\begin{theorem}
Let $w\not\in \{a^{n-1}b,ab^{n-1},ba^{n-1},b^{n-1}a\}$ and $L=\Sigma^{*}w\Sigma^{*}$, where $|w|=n$. Then $\sigma(L)=n^2+1+N(w)$.
\end{theorem}

\begin{proof}
Build the minimal automaton $\mathscr{A}_L$ recognizing $L$. We refer to words $u\in \Sigma^{*}$ as pairs $(s,p)$, where $s$ is the maximal suffix of $w$ that appears in $u$ as a prefix, and $p$ is the maximal prefix of $w$ that is also a suffix of $u$. For instance, consider the word $w=aabab$.  For the word $u=abbaba$ the corresponding pair is $(ab,a)$. Assume that $s\neq w$ and $p\neq w$.

First we show that the words corresponding to different pairs $(s_1,p_1)$ and $(s_2,p_2)$ define different transformations of the DFA $\mathscr{A}_L$. Indeed, if $p_1\neq p_2$, then the first (initial) state of $\mathscr{A}$ is mapped by this prefixes into different states ($p_1$ and $p_2$ respectively). Thus, the corresponding transformations act differently on the first state. If $s_1\neq s_2$, then without loss of generality assume that$|s_1|\leq |s_2|$. In this case there exists a prefix-state, which $s_2$ maps to the terminal state, and $s_1$ does not. Thus, again the corresponding transformations act differently on this state. If in pairs $(s_1,p_1)$ and $(s_2,p_2)$ we have that $s_1=s_2=w$ or $p_1=p_2=w$, then all words corresponding to such pairs generate the same transformation, since $w\in Syn(\mathscr{A})$.

We construct for each pair $(s,p)$ ($s$ or $p$ also can be empty words) a word in the syntactic semigroup which differs from $w$. Take the suffix $s$, append to it $2\cdot|w|-|s|$ times the letter, different from the last letter of $s$. Denote this letter by $\overline{x}$. If $s=\varepsilon$, then $\overline{x}$ is chosen to be different from the last letter of $w$. Further we append $2\cdot|w|-|p|$ times the letter, different from the first letter of $w$. Denote this letter by $\overline{y}$, then complete the word by adding $p$ to the end. By the construction $s$ is the maximal suffix of $w$ that appears in $u$ as a prefix, and $p$ is the maximal prefix of $w$, which is also a suffix of $u$. However, it could happen that after adding blocks of $\overline{x}$ and $\overline{y}$ the word $w$ appeared in $u$. It would be only in the case when $w=\overline{x}^{l}\overline{y}^{m}$ for some $l,m\ge1$. Assume that $l,m>1$, because $w\notin\{a^{n-1}b,\ ab^{n-1},\ ba^{n-1},\ b^{n-1}a\}$.  For such a word $w$ we construct $u$ as follows. First append the letter $\overline{x}$ as above. Then add the word $\overline{y}\overline{x}$, and then carry on the construction as described above. The word $w$ is not a factor of the constructed word $u$. So if $w\notin\{a^{n-1}b,\ ab^{n-1},\ ba^{n-1},\ b^{n-1}a\}$, then the syntactic semigroup of $L$ consists of at least $n^2+1$ elements.

If words $u$ referred to the pair $(s,p)$ have length at least $4n$, then these words define the same transformation of $\mathscr{A}_L$. Otherwise, applying the word $u$ to some state $q$ we may obtain another prefix $t$:
$$\overbrace{\underbrace{qu}_{t}\cdots}^{w}$$
This situation occurs only when $u$ is an inner factor of $w$. So we need to add in the transition semigroup all inner factors of $w$.

And finally the whole word $w$ always belongs to the transition semigroup.
\qed
\end{proof}

Note that for the word $a^{n-1}b$ it is impossible to construct words corresponding to pairs of the form $(\varepsilon,p)$ or $(b,p)$, where $p$ is a prefix $w$. In this case $\overline{x}=a$, $\overline{y}=b$. And the word $u$ constructed by the algorithm from the Theorem contains $w$ as a factor. Hence we need to find the words corresponding to the pairs of this form separately. There is no word in the transition semigroup corresponding to the pair $(\varepsilon,\varepsilon)$. Pairs $(\varepsilon,p)$ can be associated with the corresponding prefixes $p$, and pairs $(b,p)$ with the word $bp$. All factors of the word has already been considered, because they coincide with prefixes and suffixes. Finally, $\sigma(L)=1+n^2-2n+2n-1=n^2$.

Analogously, the same result holds for words $ab^{n-1}$, $b^{n-1}a$, and $ba^{n-1}$. The number of all different inner factors is estimated as $N(w)\leq\frac{(n-1)(n-2)}{2}$.
And we have the following estimation: $n^2\leq \sigma(L)\leq 1.5n^2+o(n^2)$. In particular, lower bound is tight. Moreover, the equality $\sigma(L)=n^2$ holds only for words $a^{n-1}b$, $ab^{n-1}$ and $b^{n-1}a$, $ba^{n-1}$.

The upper bound of the value $\sigma(L)$ is $1.5n^2+o(n^2)$. Next we give an example of the language $L_w=\Sigma^{*}w\Sigma^{*}$ for which the equality $\sigma(L_w)=1.5n^2+o(n^2)$ takes place.

\begin{proposition}
\label{upbound}
There exists a word $w$ of length $|w|=n\geq 21$ for which $\sigma(L_w)= 1.5n^2+o(n^2)$.
\end{proposition}
\begin{proof}
We prove the Proposition in a constructive way. Take the word $w=ab^{2}a^{3}b^{4}\cdots a^{k-1}b^{k}$, i.e. $n=1+2+\cdots+k$ for even $k\geq 4$. Count $N(w)$, the number of all different inner factors $v$ of the word $w$. Denote by $m$ the maximal power of $b$ that appears in $v$, i.e. $v=tb^{m}r$.

If $m=0$, then $v$ does not contain $b$, there are $k-1$ such factors.

If $m=1$, then we have three cases:

{\par}{\par}$v=ba^{l}$ ($0\leq l\leq k-1$);
{\par}{\par}$v=a^{l}b$ ($1\leq l\leq k-1$);
{\par}{\par}$v=ba^{l}b$ ($l$ is odd from $3$ to $k-1$).

There are $k+k-1+\frac{k-2}{2}=\frac{5}{2}k-2$ such factors.

If $m=2$, then we have three cases:

{\par}{\par}$v=b^2a^{l}$ ($0\leq l\leq k-1$);
{\par}{\par}$v=a^{l}b^2$ ($1\leq l\leq k-1$);
{\par}{\par}$v=b^{2}a^{l}b$, or $v=ba^{l}b^{2}$, or $v=b^{2}a^{l}b^{2}$ ($l$ is odd from $3$ to $k-1$).

There are $k+k-1+3\frac{k-2}{2}=\frac{7}{2}k-4$ such factors.

If $m=k-1$, then $v=tb^{k-1}$, where $0\leq |t|\leq n-k-1$. There are $n-k$ such factors.

Consider the case of odd $m$. Let $2<m<k-1$. We have the following cases:

{\par}{\par}$v=tb^{m}$, $0\leq |t|\leq 2+3+\cdots+m=\frac{(2+m)(m-1)}{2}$;
{\par}{\par}$v=b^{m}a^{l}$, $1\leq l\leq k-1$;
{\par}{\par}$v=b^{m}a^{l}b^{r}$, $1\leq r\leq m$, $m+2\leq l\leq k-1$, $l$ is odd;
{\par}{\par}$v=a^{l}b^{m}$, $m+1\leq l \leq k-1$;
{\par}{\par}$v=b^{r}a^{l}b^{m}$, $1\leq r\leq m-1$, $m+2\leq l\leq k-1$, $l$ is odd.

There are $\frac{m(m+1)}{2}+k-1+\frac{k-m-1}{2}m+k-1-m+\frac{k-m-1}{2}(m-1)$ such factors.

Consider the case of even $m$. Let $2<m<k-1$. The following cases are possible:

{\par}{\par}$v=tb^{m}$, $0\leq |t|\leq 2+3+\cdots+m-1=\frac{(1+m)(m-2)}{2}$;
{\par}{\par}$v=b^{m}a^{l}$, $1\leq l\leq k-1$;
{\par}{\par}$v=b^{m}a^{l}b^{r}$, $1\leq r\leq m$, $m+1\leq l\leq k-1$, $l$ is odd;
{\par}{\par}$v=a^{l}b^{m}$, $m\leq l \leq k-1$;
{\par}{\par}$v=b^{r}a^{l}b^{m}$, $1\leq r\leq m-1$, $m+1\leq l\leq k-1$, $l$ is odd;
{\par}{\par}$v=tb^{m}a^{l}$, $1\leq |t|\leq 2+3+\cdots+m-1=\frac{(1+m)(m-2)}{2}$ and $1\leq l \leq m+1$;
{\par}{\par}$v=tb^{m}a^{m+1}b^{r}$, $1\leq |t|\leq 2+3+\cdots+m-1=\frac{(1+m)(m-2)}{2}$ and $1\leq r \leq m$.

There are $\frac{m(m-1)}{2}+k-1+\frac{k-m}{2}m+k-m+\frac{k-m}{2}(m-1)+\frac{(m+1)^2(m-2)}{2}(m-1)+\frac{m(m+1)(m-2)}{2}$ such factors.

Finally, the total number of all different inner factors is equal to
$$
\begin{array}{l}
\displaystyle{
N(w)=k-1+\frac{5}{2}k-2+\frac{7}{2}k-4+n-k+ \sum^{k-3}_{m=3,\, m \text{ is odd}}\left(\frac{m(m+1)}{2}+\right.}\\
\smallskip
\displaystyle{\left.+k-1+\frac{k-m-1}{2}m+k-1-m+\frac{k-m-1}{2}(m-1)\right)+}\\
\smallskip
\displaystyle{+\sum^{k-2}_{m=4,\, m \text{ is even}}\left(\frac{m(m-1)}{2}+k-1+\frac{k-m}{2}m
+k-m+\frac{k-m}{2}(m-1)+\right.}\\
\smallskip
\displaystyle{+\frac{(m+1)^2(m-2)}{2}(m-1)+\frac{m(m+1)(m-2)}{2})=n+6k-7+}\\
\smallskip
\displaystyle{+\underbrace{\sum^{k-3}_{m=3}\left(\frac{m(m-1)}{2}+\frac{k-m}{2}(2m-1)+2k-m-1\right)}_{N_1(k)}+}\\
\smallskip
\displaystyle{+\underbrace{\sum^{k-3}_{m=3,m-\text{odd}}\left(-\frac{1}{2}\right)}_{N_2(k)}+\underbrace{\sum^{k-2}_{m=4,m-\text{even}}\left(m^3-\frac{m^2}{2}-\frac{5m}{2}-1\right)}_{N_3(k)}.}
\end{array}
$$
Count the values $N_1(k)$, $N_2(k)$, $N_3(k)$. It is easy to find $N_2(k)=-\frac{1}{2}\cdot\frac{k-4}{2}=-\frac{k-4}{4}$. Then using formula $2^3+4^3+\cdots+(2n)^3=2n^2(n+1)^2$ and $2^2+4^2+\cdots+(2n)^2$ obtain $N_3(k)=\frac{k^4}{8}-\frac{7}{12}k^3+\frac{k^2}{8}+\frac{7}{12}k+1$.
Finally, using the formula $1^2+2^2+\cdots+n^2=\frac{1}{3}n^2+\frac{1}{2}n^2+\frac{1}{6}n$, find $N_1(k)=\frac{k^3}{3}+\frac{k^2}{4}-\frac{103}{12}k+9$. Generalizing all results obtain $\sigma(L)=n^2+1+n+6k-7+N_1(k)+N_2(k)+N_3(k)=n^2+n+\frac{k^4}{8}-\frac{k^3}{4}+\frac{3}{8}k^2-\frac{9}{4}+5$. Note that $n=1+2+\cdots+k=\frac{k^2+k}{2}$. Then $\sigma(L)=\frac{3}{2}n^2+\frac{5}{2}n-kn-3k+5$, where $k$ is a positive root of the equation $k^2+k-2n=0$.

Let $n\geq 21$ can not be decomposed into a sum $n=1+2+\cdots+k$ for some even $k$. Find the number $k$, such that $1+2+\cdots+k<n\leq1+2+\cdots+k+1$. Consider the word $w=ab^{2}a^{3}\ldots b^{k}a^{l}$, where $l=n-\frac{1+k}{2}k$. Otherwise, $1+2+\cdots+k+1<n<1+2+\cdots+k+2$, then construct the word $w=ab^{2}a^{3}\ldots b^{k}a^{k+1}b^{r}$, where $r=n-\frac{2+k}{2}(k+1)$. Using analogous arguments one may verify that the equality $\sigma(L)= 1.5n^2+o(n^2)$ holds.
\qed
\end{proof}

The previous theorem provides rather simple formula to calculate the syntactic complexity of a principal ideal language. Indeed, we do not need to construct the minimal automaton of the language and then analyze its transition semigroup for an arbitrary $w$. We just need to count all different inner factors of $w$ and it can be done in time $O(n^2)$ in a trivial way.

\textbf{Acknowledgement}. The authors acknowledge support from the Presidential Programm for young researchers, grant MK-266.2012.1.


\begin{thebibliography}{99}

\bibitem{BrzId}
Brzozowski J., Ye Y.\emph{Syntactic Complexity of Ideal and Closed Languages.} In G. Mauri, A. Leporati (Eds.) Proc. DLT 2011, Lect. Notes Comp. Sci. Vol. 6795, Springer-Verlag Berlin-Heidelberg 2011. P. 117--128.

\bibitem{Ce64}
J.~\v{C}ern\'{y}. \emph{Pozn\'{a}mka k homog\'{e}nnym eksperimentom s kone\v{c}n\'{y}mi
automatami.}, Mat.-Fyz.\ Cas.\
Slovensk.\ Akad.\ Vied.
\textbf{ 14} (1964) 208--216.


\bibitem{HolK3}
Holzer, M., K\"{o}nig, B.: \emph{On deterministic finite automata and syntactic monoid size.}
Theoret. Comput. Sci. 327 (2004) P. 319--347

\bibitem{SOFSEM}
Maslennikova M.I. \emph{Reset Complexity of Ideal Languages.} In M. Bielikov\'{a}, G. Friedrich, G. Gottlob, S. Katzenbeisser, R. \v{S}p\'{a}nek, G. Tur\'{a}n (eds.) Int. Conf. SOFSEM 2012, Proc. Volume II, Institute of Computer Science Academy of Sciences of the Czech Republic, 2012, P. 33--44.

\bibitem{Myhill}
Myhill, J.: \emph{Finite automata and representation of events}. Wright Air Development
Center Technical Report, 57–624 (1957)

\bibitem{Perrin}
D. Perrin \emph{Finite automata.} Handbook of Theoretical computer Science, J. van Leewen, (ed.), Elsevier, B., P. 1--57, 1990.

\bibitem{PribR1}
E. Pribavkina, E. Rodaro. \emph{Synchronizing automata with finitely
 many minimal synchronizing words}// Inf. and Comput.
 V.\textbf{209}(3), 2011. P.568--579.

\bibitem{PribR2}
E. Pribavkina, E. Rodaro. \emph{Recognizing synchronizing automata
 with finitely many minimal synchronizing words is PSPACE-complete}//
 B. L\"{o}we, D. Normann, I. Soskov, A. Soskova (Eds.) Proc. CiE 2011,
 Lect. Notes Comp. Sci. Vol. 6735, Springer-Verlag Berlin-Heidelberg
 2011. P. 230--238.

\bibitem{Vo_Survey}
M.\,V.~Volkov. \emph{Synchronizing automata and the \v{C}ern\'y
conjecture}, in C.\,Mart\'\i{}n-Vide, F.\,Otto, H.\,Fernau (eds.), Languages
and Automata: Theory and Applications. LATA 2008. Lect.\ Notes Comp.\ Sci.
\textbf{ 5196}, Berlin, Springer (2008) 11--27.

\bibitem{Vo_CIAA07}
M. V. Volkov \emph{Synchronizing automata preserving a chain of partial orders}//
In J. Holub and J. \v{Z}d\'arek (eds.) Implementation and Application of Automata. Proc.\ 12th Int.\ Conf.\ CIAA 2007, Lect. Notes Comp. Sci., Springer-Verlag, Berlin-Heidelberg-New York. 2007. V.{4783}. P.{27--37}.

\end{thebibliography}
\end{document}